\documentclass[reqno]{amsart}
\usepackage{amssymb,graphics,graphicx,bbm,float,lscape}

\usepackage{tikz}
\usetikzlibrary{arrows,automata}
\usetikzlibrary{matrix,arrows}

\def\be{\begin{equation}}
\def\ee{\end{equation}}

\numberwithin{equation}{section}
\newtheorem{theorem}{Theorem}[section]
\newtheorem{proposition}[theorem]{Proposition}
\newtheorem{lemma}[theorem]{Lemma}
\newtheorem{corollary}[theorem]{Corollary}
\newtheorem{assumption}{Assumption}
\newtheorem{definition}{Definition}
\newtheorem{remark}{Remark}

\begin{document}
\title{Virial Expansion Bounds Through Tree Partition Schemes}

\author[S. Ramawadh]{Sanjay Ramawadh}
\address{Department of Mathematics, Utrecht University, P.O. Box 80010, 3508 TA Utrecht}
\email{s.d.ramawadh@uu.nl}

\author[S. J. Tate]{ Stephen James Tate}
\address{Department of Chemical Engineering, South Kensington Campus, Imperial College London, SW7 2AZ, U.K.}
\address{ Ruhr-Universit\"{a}t Bochum, Fakult\"{a}t f\"{u}r Mathematik, 44870 Bochum, Germany}
\address{Department of Mathematics, University of Sussex, Falmer Campus, Brighton BN1 9QH, U.K.}
\email{s.tate@imperial.ac.uk}



\begin{abstract}
In this paper, we use tree partition schemes and an algebraic expression for the virial coefficients in terms of the cluster coefficients in order to derive upper bounds on the virial coefficients and consequently lower bounds on the radius of convergence of the virial expansion $\mathcal{R}_{\textrm{Vir}}$. The bound on the radius of convergence in the case of the Penrose partition scheme is the same as that proposed by Groeneveld and improves the bound achieved by Lebowitz and Penrose. 
\end{abstract}

\thanks{\copyright{} 2015 by the authors. This paper may be reproduced, in its entirety, for non-commercial purposes.}

\maketitle

\section{Introduction}

In this paper, we derive virial expansion bounds for stable potentials, based on the tree partition scheme of Penrose \cite{P67}. These bounds offer an improvement on the radius of convergence of the virial expansion $\mathcal{R}_{\textrm{Vir}}$ when compared to those of Lebowitz and Penrose \cite{LP64} and are actually the same as those proposed by Groeneveld \cite[Chapter IV, Section 3.7]{G67}. Towards achieving these bounds, we introduce a new approach to understanding the virial coefficients as a weighted sum over tree graphs. This offers an alternative approach to virial coefficients than those made in the literature.

The virial expansion provides a description of the relationship between macroscopic properties of a system of particles, specifically pressure and density. The virial expansion may be viewed as a generalisation of the ideal gas law $V P = NkT$. $V$ denotes the volume in which the system is contained, $P$ is the pressure of the system, $N$ is the number of particles in the system, $k$ is Boltzmann's constant and $T$ is the temperature of the system. 

The idea of refining the ideal gas law based on an expansion of pressure as a power series in density was proposed by Thiessen \cite{T85} and was first undertaken by Kamerlingh Onnes in 1901 \cite{K01}. An important reference for cluster and virial expansions is the book of Mayer \cite{MM40}, who successfully derived the interpretation of the virial coefficients in terms of weighted two-connected graphs. Much work was done on the virial expansion in the 1950s and 60s, including: the development of Kirkwood-Salsburg Equations \cite{KS53}; the graphical approximation approach of Groeneveld \cite{G62}; deriving inequalities for the thermodynamic functions \cite{L63,P63}; and bounds obtained from using Lagrange inversion on the cluster expansion \cite{LP64}.

In the literature, virial expansions have undergone a recent revival. There are many different approaches to the virial expansion. These include: a modification of the Lebowitz-Penrose argument \cite{T13}; derivations through the Canonical Ensemble \cite{PT12, MP13}; and using differential equations for functions related to the pressure with a uniform repulsive potential \cite{BM14}.  There is also the recent paper of Morais, Procacci and Scoppola \cite{MPS14} for the Lennard-Jones gas. Furthermore, the recent work of Jansen \cite{J12}, provides physical interpretations for the different radii of convergence, emphasising cases in which the virial expansion converges for a larger range of densities than the cluster expansion. 

A tree partition scheme for connected graphs was first provided by Penrose \cite{P67}. Sokal \cite{S01} and Fern\'{a}ndez and Procacci \cite{FP07} give a more general version and apply it to the Potts model and hardcore gases respectively. This paper develops this partition scheme for two-connected graphs.

Providing a weighted tree expression for the virial coefficients should also have consequences for numerical results that have been obtained. The number of tree graphs grows at a smaller rate than those of two-connected graphs. Indeed, by Cayley's formula, we have $n^{n-2}$ trees on $n$ vertices compared to at least $2^{{n \choose 2} -n}$ two-connected graphs on $n$ vertices. The lower bound arises from considering the graph corresponding to a cycle on $n$ vertices, which is certainly two-connected. We may add or not any of the remaining ${n \choose 2}-n$ edges in the graph to give some subset of two-connected graphs. The work of Ree-Hoover \cite{RH64} established a method of reducing the number of two-connected graphs that one needs to consider by introducing a factor $1=h_{i,j}-f_{i,j}$ for every missing edge in each two-connected graph and finding the cancellations that arise from this. This has led to Monte-Carlo simulations to give many virial coefficients in the work of Clisby and McCoy \cite{CM05}. A practical application of these expressions would be to evaluate the tree versions provided by the expression of virial coefficients in terms of weighted trees.

In future work, we hope to apply this method to other partition schemes in particular to the scheme indicated by the notion of internally and externally active edges of a matroid \cite{ S01, B08}, which is given explicitly in \cite{T14}. Furthermore, the hardcore cases in \cite{FP07} will give good further examples to apply the tree partition scheme to.

The paper is structured so that in section \ref{sec:clustertovirial}, we give a derivation for the relationship between cluster and virial coefficients emphasising their combinatorial properties. The results of this paper are summarised in section \ref{sec:results} and we prove that this result  is the same as that proposed by Groeneveld in section \ref{sec:groeneveldpenrose} . In section \ref{sec:partitionschemes}, we provide the general framework in which we understand the tree partition schemes and the important properties. We apply this to the Penrose partition scheme in section \ref{sec:penrosepartitionscheme} for positive and stable potentials and present the conclusions and further work in section \ref{sec:conclusion}.

\section{The Classical Gas Model} \label{sec:clustertovirial} 

The results for this paper are for the classical many particle system.

The canonical partition function for a fixed number of particles is given by:
\be Z_N:=  \frac{1}{N!}\prod\limits_{i=1}^N\left(\int_{\mathbb{R}^d}\, \mathrm{d}^dx_i\right) e^{-\beta U_N(x_1, \cdots, x_N)}, \ee 
where $U_N$ is the interaction potential. We make important assumptions on the interaction potential in order to achieve the results, which are given in subsection \ref{subsec:potentialassumptions}. In subsection \ref{subsec:weightedgenfns} we indicate the connection between cluster coefficients and weighted connected graphs and in subsection \ref{subsec:virialcluster}, we derive the algebraic relationship.

\subsection{Assumptions on the potential}
\label{subsec:potentialassumptions}
\begin{assumption}[Potential]
The $N$-particle interaction potential: $U_N(x_1, \cdots, x_N)$ may be written as the sum of pair-potentials:
\be U_N(x_1, \cdots, x_N)= \sum\limits_{1 \leq i<j \leq N}\Phi(x_i,x_j). \ee
Furthermore, we assume that the pair potentials $\Phi(x_i,x_j)$ are central, that is, they only depend on the distance from $x_i$ to $x_j$. 
\end{assumption}

\begin{definition}[Stability]
A pair potential is called \emph{stable} if for all $N \in \mathbb{N}$ and all collections of locations $(x_1, \cdots, x_N) \in \mathbb{R}^{dN}$, we have the following inequality:

\be \sum\limits_{1 \leq i<j \leq N}\Phi(x_i,x_j) \geq -BN. \ee 

\end{definition}
\begin{definition}[$C(\beta)$ - `temperedness']
The `temperedness' function $C(\beta)$ plays an important r\^{o}le in the estimates of convergence through tree partition schemes.
\be C(\beta):=\int\limits_{\mathbb{R}^d} \, \left|e^{-\beta \Phi(0,x)}-1\right| \mathrm{d}^dx \label{Cbeta}. \ee
If the expression \eqref{Cbeta} is finite then the potential $\Phi$ is called `tempered', which we assume for bounds involving $C(\beta)$.
\end{definition}

\subsection{The Classical Gas and Weighted Generating Functions}
\label{subsec:weightedgenfns}

We use the finite-volume grand canonical partition function for a classical gas:
\be \Xi_\Lambda(z)  = \sum\limits_{N=0}^{\infty}Z_{N,\Lambda}z^N \ee
where, given the assumptions on the potential, $Z_{N,\Lambda}$ may be expressed as:
\be Z_{N,\Lambda} = \frac{1}{N!}\prod\limits_{i=1}^N\left(\int_{\Lambda^N}\, \mathrm{d}^dx_i\right) e^{-\beta\sum\limits_{1 \leq i<j \leq N}\Phi(x_i,x_j)}. \ee
We have the relationship $\beta P_\Lambda = \frac{1}{|\Lambda|} \ln \Xi_\Lambda$, which gives us an expansion for the finite-volume pressure. 

The cluster and virial expansions for the finite-volume pressure are expressed as:

\be
\beta P_\Lambda = \sum_{n=1}^\infty \frac{b_n(\Lambda)}{n!}z^n, \label{eq:cluster-finitevol}
\ee
\be
\beta P_\Lambda = \sum_{n=1}^\infty \frac{\beta_n(\Lambda)}{n!} \rho^n. \label{eq:virial-finitevol}
\ee

The variable $z=e^{\beta \mu}$ is the activity or fugacity of the system and gives control over the average particle number in the grand canonical ensemble. The variable $\mu$ is the chemical potential and $\beta=\frac{1}{kT}$ is the inverse temperature. $\rho$ denotes the density of the system of particles. 
The thermodynamic pressure $p$ is then found by taking the limit $\Lambda\uparrow\mathbb{R}^d$ in the sense of Fisher. The cluster and virial expansion for the thermodynamic pressure are expressed as:
\be
p = \sum_{n=1}^\infty \frac{b_n}{n!}z^n, \label{eq:cluster}
\ee
\be
p = \sum_{n=1}^\infty \frac{\beta_n}{n!} \rho^n. \label{eq:virial}
\ee
The fugacity and the density are related to each other as follows:
\be
\rho = z \frac{dp}{dz} = \sum_{n=1}^\infty \frac{b_n}{(n-1)!}z^n. \label{eq:density} 
\ee

Derivations and expressions of the coefficients $b_n(\Lambda)$, $\beta_n(\Lambda)$, $b_n$, and $\beta_n$ in terms of connected and two-connected graphs are found in the book of Mayer \cite[Chapter 13]{MM40}.

Mayer's trick is to introduce $f_{i,j}:=\exp(-\beta \Phi(x_i,x_j))-1$ and to rewrite the integrand as:
\begin{align*}
\exp(-\beta \sum_{1 \leq i<j \leq N} \Phi(x_i,x_j)) &= \prod\limits_{1 \leq i<j \leq N} \exp(-\beta \Phi(x_i,x_j)) \\
&= \prod\limits_{1 \leq i<j \leq N} (1+f_{i,j}) \\
&= \sum\limits_{g \in \mathcal{G}[N]}\prod\limits_{e \in E(g)} f_e, \end{align*}
where $\mathcal{G}[N]$ denotes the collection of all graphs on the vertex set $[N]$.

We may define the graph weight:
\be w(g):= \prod\limits_{e \in E(g)} f_e, \ee
then the canonical partition function may be written as:
\be Z_{N,\Lambda}= \frac{1}{N!}\prod\limits_{i=1}^N\left(\int_{\Lambda}\, \mathrm{d}^dx_i\right) \sum\limits_{g \in \mathcal{G}[N]} w(g) \ee
and we realise that the grand canonical partition function may be written as the weighted generating function of graphs. One may also use the weights including the integrals as:
\be 
W_\Lambda(g) = \prod\limits_{i \in V(g)} \left(\int_{\Lambda} \, \mathrm{d}^dx_i \right) w(g) .
\ee

Mayer derived that the pressure can be written as the weighted exponential generating function of connected graphs, that is, for the weight:
\be
\tilde{w}(g):= \lim\limits_{\Lambda \uparrow \mathbb{R}^d}\frac{1}{|\Lambda|}W_\Lambda(g). \ee 

We have that:
\be
b_n = \sum\limits_{g \in \mathcal{C}[n]} \tilde{w}(g), \ee
where $\mathcal{C}[n]$ denotes the collection of connected graphs on $n$ labels.

\subsection{Deriving the Algebraic Relationship}
\label{subsec:virialcluster}

It is possible to relate the virial coefficients $\beta_n$ to the well understood cluster coefficients $b_n$. This is done by inverting \eqref{eq:density}, so as to obtain an expression for the fugacity as a power series in terms of the density, and then substituting this power series into \eqref{eq:cluster}. Lagrange inversion \cite{C74} gives the following:

\be
\beta_{n+1} = n! [z^n] \left[\left(\frac{\rho(z)}{z}\right)^{-n}\right], \label{eq:Lagrange}
\ee

Using this expression, we can find a nice combinatorial expression of the $\beta_n$ in terms of the $b_n$.

Define, for any $x \in \mathbb{R}$ and any nonnegative integer $k$, the falling factorial as: 

\be
(x)^{\underline{k}} := x(x-1)\cdots (x-k+1) = \prod_{i=0}^{k-1} (x-i).
\ee

If $x$ is a positive integer, we may also write $(x)^{\underline{k}} = \frac{x!}{(x-k)!}$ and $\frac{(x)^{\underline{k}}}{k!}=\binom{x}{k}$. We use this last equality to extend the definition of binomial coefficients. For $x\in \mathbb{R}$ and any nonnegative integer $k$, we define 
\be
\binom{x}{k} := \frac{(x)^{\underline{k}}}{k!}.
\ee

Let $x_1,x_2,\ldots$ be an infinite number of variables, then the partial Bell polynomials $B_{n,k}\left(x_1,x_2,\ldots,x_n\right)$ are defined by:

\be
\exp\left( u \sum_{m=1}^\infty x_m \frac{t^m}{m!} \right) = 1 + \sum_{n=1}^\infty \frac{t^n}{n!} \left(\sum_{k=1}^n u^k B_{n,k}\left(x_1,\ldots, x_n\right)\right).
\ee

There exists an explicit expression for the partial Bell polynomials:

\be
B_{n,k}\left(x_1,x_2,\ldots,x_n\right) = \sum_{{(k_1,\ldots,k_n): k_i\geq 0}\atop{\sum_i k_i = k, \sum_i ik_i = n}} \frac{n!}{\prod_{i=1}^n k_i!} \prod_{i=1}^{n} \left(\frac{x_i}{i!}\right)^{k_i}.
\ee

For more known results on Bell polynomials we refer the reader to \cite{C74}. Lastly, we define the potential polynomials. Let $x_0=1$, let $x_1,x_2,\ldots$ be an infinite sequence of numbers and let $r$ be any number. The potential polynomials $P_n^{(r)}\left(x_1,\ldots,x_n\right)$ are defined by:

\be
\left(\sum_{n=0}^\infty x_n\frac{t^n}{n!}\right)^r = 1 + \sum_{n=1}^\infty P_n^{(r)}\left(x_1,\ldots,x_n\right)\frac{t^n}{n!}. \label{eq:potential1}
\ee

There is also the following explicit expression for the potential polynomials:

\be
P_n^{(r)}\left(x_1,\ldots,x_n\right) = \sum_{k=1}^n (r)^{\underline{k}} B_{n,k}\left(x_1,\ldots,x_n\right). \label{eq:potential2}
\ee

We return to \eqref{eq:Lagrange}. Since $\frac{\rho(z)}{z}=\sum_{n=1}^\infty \frac{b_n}{(n-1)!}z^{n-1}=\sum_{n=0}^\infty b_{n+1}\frac{z^n}{n!}$ and $b_1=1$, it follows from \eqref{eq:potential1} and \eqref{eq:potential2} that:

\begin{align}
\beta_{n+1} &= n! \frac{P_n^{(-n)}\left(b_2,\ldots,b_{n+1}\right)}{n!}\nonumber \\ 
 &= P_n^{(-n)}\left(b_2,\ldots,b_{n+1}\right)\nonumber\\
 &= \sum_{k=1}^n (-n)^{\underline{k}} B_{n,k}\left(b_2,\ldots,b_{n+1}\right)\nonumber \\
 &= \sum_{k=1}^n \binom{-n}{k} k! B_{n,k}\left(b_2,\ldots,b_{n+1}\right). \label{eq:virial-Bell}
\end{align}

\section{The Main Result}
\label{sec:results}
The main result of this paper is the improved radius of convergence of the virial expansion. En route to this improved radius of convergence, we obtain a new expression for the virial coefficients in terms of weighted tree graphs. This new expression is given in section \ref{sec:penrosepartitionscheme}.

For a stable interaction, we let $u:=\exp(2 \beta B)$, where $B$ is the stability constant. We have the following bound on the radius of convergence of the virial expansion, $\mathcal{R}_{\mathrm{Vir}}$, achieved through using tree partition schemes for connected graphs, arising from the Penrose partition scheme. The proof for the following theorem is given in section \ref{sec:penrosepartitionscheme}. The function $T_1(z)$ is the following generating function:
\be
T_1(z) =\sum\limits_{n=1}^{\infty}n^n\frac{z^{n+1}}{(n+1)!}+z. \label{eq:T1def}
\ee

{
\renewcommand{\thetheorem}{\ref{thm:stablebounds}}
\begin{theorem}
For a classical gas with a stable potential, the virial coefficients have the following upper bound:
\be \frac{|\beta_{n+1}|}{(n+1)!} \leq \frac{C(\beta)^{n-1}}{n} \left(\frac{1+\frac{1}{u}T_1(uc(u))}{c(u)}\right)^n, \ee
where $c(u)$ is the smallest positive root of 
\be c(u)T_1'(uc(u))-\frac{1}{u}T_1(uc(u))=1. \ee
\end{theorem}
\addtocounter{theorem}{-1}
}
This result is precisely the same as the one proposed by Groeneveld \cite[Chapter IV, Section 3.7]{G67}, which is proved in section \ref{sec:groeneveldpenrose}.

\section{Comparison of the bound to previous bounds and application to particular models}
\label{sec:examples}
We have the exact equation \eqref{eq:exactpenrose} for the virial coefficients, which may be used to calculate the coefficients in particular examples. One just has to calculate the weights of non-splittable Penrose trees.

For stable potentials we have a bound that is an improvement of Lebowitz Penrose\cite{LP64} that we can see in in figure \ref{fig:comparison}. In particular for positive potentials, we obtain $\mathcal{R}_{\mathrm{Vir}} \geq 0.237961 C(\beta)^{-1}$, whereas the Lebowitz and Penrose bound gives $\mathcal{R}_{\mathrm{Vir}} \geq 0.144766998 C(\beta)^{-1}$.

\begin{figure}[here]\begin{center}
\rotatebox{-1}{\includegraphics[width=0.3\textheight,angle=-89]{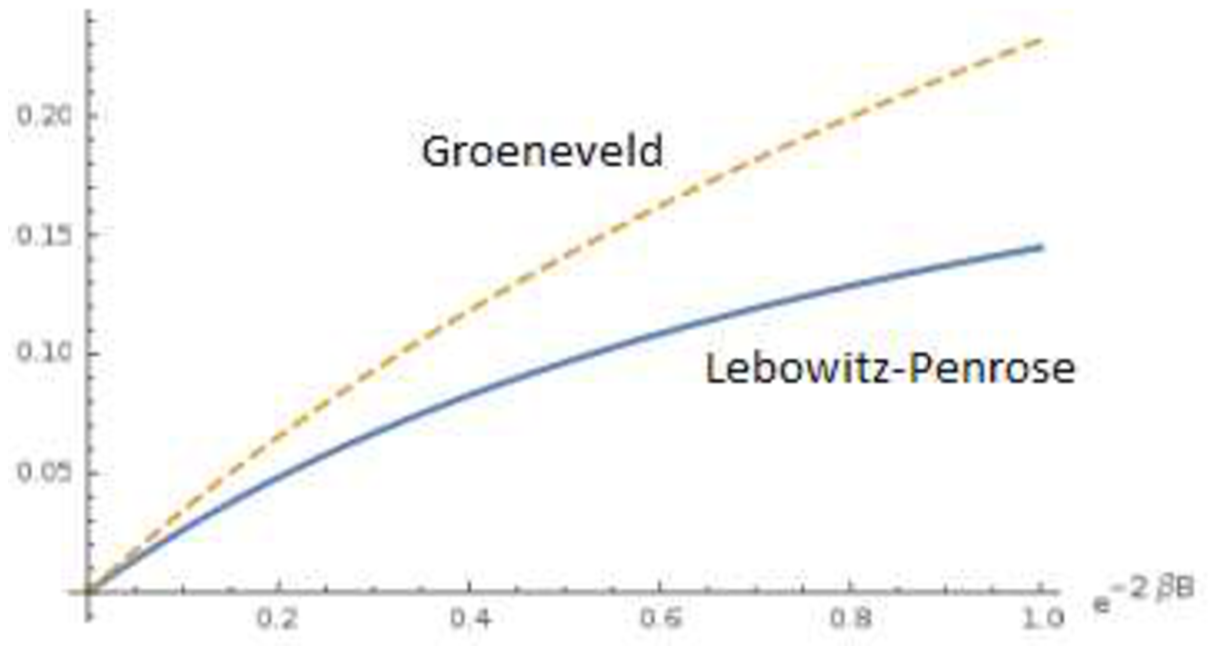}}
\caption{A comparison of the lower bound  for the radius of convergence of the virial expansion obtained by Lebowitz-Penrose with that given by Groeneveld \label{fig:comparison}}
\end{center}\end{figure}

\section{Penrose is the same as Groeneveld}
\label{sec:groeneveldpenrose}

In \cite{G67}, Groeneveld obtains, for stable potentials, the bound $\mathcal{R}_{\mathrm{Vir}} > \alpha(u)C(\beta)^{-1}$, where $\alpha(u)$ is the smallest positive root of $\alpha(u)e^{-\alpha(u)}=\frac{1}{(1+u)e}$. His method relies upon finding recursion relations for graphs, expressing these as a system of differential equations in three variables for multivariable generating functions and solving these equations in order to obtain an expression for the virial coefficients. Here we indicate how our result actually gives the same. 

In order to show that the two bounds are equivalent, we use the fact that the rooted tree generating function is the inverse of $s \mapsto se^{-s}$ , which can be found in \cite{CGHJK96}. The function $T_1(z)$ as defined in \eqref{eq:T1def} is shown to be equal to $T_{\mathrm{Pen},1}(z)$ in section \ref{sec:penrosepartitionscheme} and we use in this section two identities proved in section \ref{sec:penrosepartitionscheme} for $T_{\mathrm{Pen},1}(z)$, but we write these in terms of $T_1(z)$ to maintain consistency with the notation of the previous section.

\begin{lemma}
The derivative of the generating function $T_1(z)$, evaluated at $se^{-s}$ satisfies the following identity:
\be T_{1}'(se^{-s})=\frac{1}{1-s}. \label{eq:relation1} \ee
\end{lemma}
\begin{proof}
We use the relation \eqref{eq:penasrootedtree}: $T_{1}'(z)=\frac{1}{1-T^{\bullet}(z)}$. If we change variables $z \mapsto se^{-s}$ and use that $T^{\bullet}(se^{-s})=s$, then we easily obtain:
\be T_{1}'(se^{-s})=\frac{1}{1-s}. \ee
\end{proof}

\begin{lemma}
The generating function $T_1(z)$ evaluated at $se^{-s}$ satisfies:
\be T_{1}(se^{-s})=1-e^{-s}. \label{eq:relation2} \ee
\end{lemma}

\begin{proof}
We know from \eqref{eq:T1inT} that $T_{1}(z)=1-\frac{z}{T^{\bullet}(z)}$, making the substitution $z \mapsto se^{-s}$, we immediately get:
\be T_{1}(se^{-s})=1-e^{-s} . \ee
\end{proof}

\begin{proposition}
The lower bound of the radius of convergence of the virial expansion:
\be \mathcal{R}_{\mathrm{Vir}} \geq C(\beta)^{-1} \frac{c(u)}{1+\frac{1}{u}T_{1}(uc(u))} \ee
where $c(u)$ is the smallest positive solution to  $uc(u)T_{1}'(uc(u))-T_{1}(c(u))=1$ is equal to $\alpha(u)C(\beta)^{-1}$, where $\alpha(u)$ is the smallest positive solution to $\alpha(u)e^{-\alpha(u)}=\frac{1}{(1+u)e} $.
\end{proposition}
\begin{proof}
Reparametrise $uc(u)=t(u)e^{-t(u)}$ and substitute into the defining equation for $c(u)$:
\be 
t(u)e^{-t(u)}T_{1}'(t(u)e^{-t(u)})-T_{1}(t(u)e^{-t(u)})=u.
\ee
We use \eqref{eq:relation1} and \eqref{eq:relation2} to simplify the expression to:
\begin{align}
\frac{t(u)e^{-t(u)}}{1-t(u)}-1 + e^{-t(u)} &= u, \notag \\
\frac{e^{-t(u)}}{1-t(u)} &=1+u. \label{eq:urelation}
\end{align}
We transform $\alpha(u)=1-t(u)$ and obtain:
\begin{align*} 
(1+u)\alpha(u)e&=e^{\alpha(u)},\\
\alpha(u)e^{-\alpha(u)} &= \frac{1}{(1+u)e}.\end{align*}
For the bound, we make the same substitution:
\be \mathcal{R}_{\mathrm{Vir}} \geq \frac{t(u)e^{-t(u)}}{u+T_{1}(t(u)e^{-t(u)})}, \ee
and use \eqref{eq:relation2} to obtain
\be \mathcal{R}_{\mathrm{Vir}} \geq \frac{t(u)}{(1+u)e^{t(u)}-1}. \ee
We use  \eqref{eq:urelation} to obtain
\be \mathcal{R}_{\textrm{Vir}} \geq \frac{t(u)}{\frac{1}{1-t(u)}-1}=1-t(u)=\alpha(u) \ee
as required.
\end{proof}

\section{Partition Schemes}
\label{sec:partitionschemes}
\label{sec:genpartitionschemes}
The improved bounds of this paper rely upon partition schemes for connected graphs in terms of tree graphs. A helpful reference for partition schemes and their application in the cluster expansion of hard-core models is found in \cite{FP07}. 

We define $[n]:=\{1, \cdots, n\}$. The collection of all connected graphs, respectively tree graphs, on the label set $[n]$ is denoted by  $\mathcal{C}[n]$,  respectively $\mathcal{T}[n]$. 

We define a partial order on $\mathcal{C}[n]$ by bond inclusion: $G \leq \tilde{G}$ $\iff$ $E(G) \subset E(\tilde{G})$. For $G \leq H$, we define the set $[G,H]=\{K\vert G \leq K \leq H\}$. We consider partitions of $\mathcal{C}[n]$, indexed by $\tau \in \mathcal{T}[n]$, comprising of blocks of the form $[\tau, R(\tau)]$, where $R: \mathcal{T}[n] \to \mathcal{C}[n]$. The mapping $R$ is called a partition scheme. More formally, we have the following definition:

\begin{definition}[Partition Scheme]
A partition scheme for a family $\mathcal{C}[n]$ of connected graphs is any map $R:\mathcal{T}[n] \to \mathcal{C}[n]$, $\tau \mapsto R(\tau)$, such that:
\begin{itemize}
\item[i)] $E(R(\tau)) \supset E(\tau)$ for all $\tau \in \mathcal{T}[n]$, and
\item[ii)] $\mathcal{C}[n]$ is the disjoint union of the sets $[\tau,R(\tau)]$ for $\tau \in \mathcal{T}[n]$.
\end{itemize}
\end{definition}

The advantage of a Boolean partition scheme is that it allows us to rewrite the connected graph sum as:

\be \sum\limits_{G \in \mathcal{C}[n]}\prod\limits_{e \in E(G)} f_e = \sum\limits_{\tau \in \mathcal{C}[n]}\prod\limits_{e \in E(\tau)}f_e \prod\limits_{\varepsilon \in E(R(\tau)) \setminus E(\tau)}(1+f_{\varepsilon}). \label{eq:scheme} \ee

This section outlines the general method of achieving bounds from partition schemes. In the next section, we apply this to the Penrose partition scheme. 

This approach requires, firstly, a well defined mapping from products of trees arising from cluster coefficients in \eqref{eq:virial-Bell} to trees on the vertex set $[n+1]$ such that the weights are preserved. This is achieved through the product structure. It is then necessary to understand the combinatorial factors that arise from this \emph{many}- to-one mapping. After achieving an exact expression, we make appropriate bounds for the tree weights, depending on whether we assume positive or stable potentials. 

\subsection{Faithful Mergings of Trees}

In expression \eqref{eq:virial-Bell} we have an expression for the virial coefficient as a many variable polynomial in the cluster coefficients. This polynomial is explicitly given by:

\be
\beta_{n+1} = \sum_{k=1}^n \binom{-n}{k} k! \sum_{{(k_1,\ldots,k_n): k_i\geq 0}\atop{\sum_i k_i = k, \sum_i ik_i = n}} \frac{n!}{\prod_{i=1}^n k_i!} \prod_{i=1}^{n} \left(\frac{b_{i+1}}{i!}\right)^{k_i}. \label{eq:virial-Bell-expl}
\ee

Equation \eqref{eq:scheme} shows how we can relate the cluster coefficients to a sum over trees using any partition scheme. Now, consider any monomial $\prod_i b_{i+1}^{k_i}$ appearing in \eqref{eq:virial-Bell-expl}. We want to find an analogous expression for this monomial in terms of trees on $n+1$ vertices. 

Using any partition scheme, the monomial $\prod_i b_{i+1}^{k_i}$ can be related to a sum of products over trees. Any term in this latter sum is a product which involves a collection of $k$ trees, say $\{\tau_1,\ldots,\tau_k\}$. We want to associate to this collection of $k$ separate trees a single tree on $n+1$ vertices. This is done by merging trees.

\begin{definition}[Merging of trees]
Let $I$ be a set of labels and let $\{\tau_1,\ldots,\tau_k\}$ be a collection of $I$-labelled trees. Let $G$ be the graph with vertex set $I$ and edge set $\{\{i,j\} \vert \, \{i,j\} \in \tau_l \, \text{for some } l\}$. The edges come with multiplicity $m(\{i,j\})=|\{l \vert \, E(\tau_l) \ni \{i,j\}\}$. Then $G$ is called the\emph{ merging} of $\{\tau_1,\ldots,\tau_k\}$.
\end{definition}

It is important to note that the merging of a set of labelled trees heavily depends on the labeling of the trees. In general, a merging may contain multiple edges between a pair of vertices, loops and cycles. It may even not be connected at all! See Figure 1 for an example.

\begin{figure}
$\begin{array}{r|l}
\begin{tikzpicture}
\tikzstyle{vertex}=[circle,draw=black,minimum size=17pt,inner sep=0pt]
	\node[vertex]	(a)	 at (1,2) {1};
	\node[vertex]	(b) at (0,0) {3};
	\node[vertex]	(c) at (2,0) {4};
	\node[vertex]	(d) at (3,2) {5};
	\node[vertex]	(e) at (3,0) {2};
	\node[vertex]	(f) at (4,0) {3};
	\node[vertex]	(g) at (5,2) {4};
	\node[vertex]	(h) at (6,0) {1};
	
	\draw (a)--(b);
	\draw (a)--(c);
	\draw (d)--(e);
	\draw (g)--(f);
	\draw (g)--(h);
\end{tikzpicture}
&
\begin{tikzpicture}
\tikzstyle{vertex}=[circle,draw=black,minimum size=17pt,inner sep=0pt]
	\node[vertex]	(a)	 at (1,2) {3};
	\node[vertex]	(b) at (0,0) {1};
	\node[vertex]	(c) at (2,0) {4};
	\node[vertex]	(d) at (3,2) {2};
	\node[vertex]	(e) at (3,0) {5};
	
	\draw (a)--(b);
	\draw (a)--(c);
	\draw (b.-10) -- (c.190);
    \draw (c.-190) -- (b.10);
	\draw (d)--(e);
\end{tikzpicture}
\end{array}$
\caption{An example of ill behaviour of mergings. On the left-hand side we have three trees labeled in $\{1,2,3,4,5\}$, and on the right-hand side we have the merging of these trees.}
\end{figure}
 
If a merging of a set of labelled trees is also a tree, we say that the set of labelled trees is \emph{properly} labelled.

Denote by $V_i$ the label set for the tree $\tau_i$ arising from a merging. 
\begin{proposition}
A necessary condition for the merging to be proper is that $|V_i \cap V_j| \leq 1$ for all $i \neq j \in [k]$. 
\end{proposition}
\begin{proof}
If $|V_i \cap V_j| \geq 2$, this means we have (at least) two vertices say $l$ and $m$ in common between $\tau_i$ and $\tau_j$. The first problem is whether we have a repeated edge in common between the two graphs, if we do then it is not properly labelled. Otherwise we have two disjoint paths from $l$ to $m$ arising from the path in $\tau_i$ and the disjoint path in $\tau_j$ since they are trees. This means that we will have a cycle in the merged graph and it therefore won't be a tree. Hence $|V_i \cap V_j| \leq 1$ is a necessary condition.
\end{proof}

We now assume that $|V_i \cap V_j| \leq 1$ for all $i \neq j \in [k]$ and provide a necessary and sufficient condition for the merging to be proper.

\begin{definition}[Merging Graph]
For a given merging with label sets $(V_i)_{i \in [k]}$, we define the merging graph $\mathcal{M}((V_i)_{i \in [k]})$ through its edge and vertex sets:
\begin{align*}
V(\mathcal{M}((V_i)_{i \in [k]}))&:=[k] \cup \left( \bigcup_{1 \leq i<j \leq k} (V_i \cap V_j) \times \{\star\}\right), \\ 
E(\mathcal{M}((V_i)_{i \in [k]}))&:=\{\{i,(V_i \cap V_j  ,\star)\} \, \{j, (V_i \cap V_j, \star)\} \vert \, V_i \cap V_j \neq \emptyset \},
\end{align*}
where $\emptyset \times \{\star\} := \emptyset$.
\end{definition}

\begin{proposition}
\label{prop:merging}
Given that the vertex sets $V_i$ satisfy $|V_i \cap V_j| \leq 1$ for all $i \neq j \in [k]$, a merging of $(\tau_1, \cdots, \tau_k)$ is properly labelled if and only if the merging graph $\mathcal{M}((V_i)_{i \in [k]})$ is a tree.
\end{proposition}

\begin{proof}
 To prove the proposition we use the two defining properties of a tree: that it is connected and acyclic.

The merging graph $\mathcal{M}((V_i)_{i \in [k]})$ is connected if and only if the merged graph $G$ is connected since there is a path between $l \in \tau_i$ and $m \in \tau_j$ in graph G  if and only if $i$ and $j$ are in the same connected component of $\mathcal{M}((V_i)_{i \in [k]})$.

We now show that there is a cycle in the merging graph $\mathcal{M}((V_i)_{i \in [k]})$ if and only if there is a cycle in the merged graph $G$. If we have a cycle in $\mathcal{M}((V_i)_{i \in [k]})$, then replacing each vertex $i$ in the cycle by the corresponding path in $\tau_i$ between its two neighbours in the cycle, we end up with a cycle in $G$. The cycle in $G$ can be transformed into a cycle in $\mathcal{M}((V_i)_{i \in [k]})$ by reducing each tree path to the corresponding vertex as no cycles will appear in an individual tree. 
\end{proof}

\begin{remark}
We realise here that the graph $\mathcal{M}((V_i)_{i \in [k]})$ is given a similar structure to the block cutpoint tree used to understand the two-connected components of connected graphs. The vertices at which we attach two trees are like the cutpoints and each individual tree is treated like a block. 
\end{remark}

Now, let $R$ be any partition scheme and consider any product $\prod_i b_{i+1}^{k_i}$ appearing in \eqref{eq:virial-Bell-expl}. As mentioned earlier, this product can be written as a sum of products of $k$ trees. Consider one such product and say it involves the collection of trees $\{\tau_1,\ldots,\tau_k\}$, then we can associate a tree on $n+1$ vertices, say $\tau$, to this collection by properly labeling the trees $\tau_1,\ldots,\tau_k$, where the labels are chosen from $[n+1]$. The tree $\tau$ is then simply the merging of $\{\tau_1,\ldots,\tau_k\}$. Since the $k$ smaller trees have $\sum_i (i+1)k_i=n+k$ vertices in total, it follows that we have to identify $k-1$ vertices. While any proper labeling with $k-1$ identifications will do, we require that the merging respects the partition scheme $R$ in the following sense:

\begin{definition}[Faithful merging]\label{def:faithfulmerging}
Let $R$ be any partition scheme and let $\{\tau_1,\ldots,\tau_k\}$ be a collection of properly $[n+1]$-labelled trees. An $R$-faithful merging of $\{\tau_1,\ldots,\tau_k\}$ is a labelled tree $\tau$ such that:
\begin{itemize}
\item[i)] $E(\tau)=\cup_{i=1}^k E(\tau_i)$, and
\item[ii)] $E\left(R(\tau)\right)\setminus E(\tau)=\cup_{i=1}^k \left(E\left(R(\tau_i)\right)\setminus E(\tau_i)\right)$.
\end{itemize}
We call $\{\tau_1,\ldots,\tau_k\}$ a splitting of $\tau$.
\end{definition}

We define the weight of a tree $\omega(\tau)$ by:

\be \omega(\tau):= \prod\limits_{e \in E(\tau)}f_e \prod\limits_{\varepsilon \in E(R(\tau))\setminus E(\tau)}(1+f_{\varepsilon}).\label{eq:weight} \ee

The importance of a faithful merging follows from the following factorisation property.

\begin{lemma}[Product Structure of Faithful Mergings]
Let $\{\tau_1,\ldots,\tau_k\}$ be given and let $\tau$ be a faithful merging of this given set of trees. We have the identity:
\be \prod\limits_{i=1}^k \omega(\tau_i) = \omega(\tau). \ee
\end{lemma}

\begin{proof}
From property i) of a faithful merging, we have that: 
\be \prod\limits_{i=1}^k \prod\limits_{e\in E(\tau_i)} f_e = \prod\limits_{e\in E(\tau)} f_e . \ee

Property ii) of the faithful merging gives that:
\be  \prod\limits_{i=1}^k\prod\limits_{\varepsilon \in E(R(\tau_i))\setminus E(\tau_i)}(1+f_{\varepsilon})= \prod\limits_{\varepsilon \in E(R(\tau))\setminus E(\tau)}(1+f_{\varepsilon}). \ee

Taking the product of the two identities, it follows that:
\begin{align}
\prod_{i=1}^k \prod\limits_{e \in E(\tau_i)}f_e \prod\limits_{\varepsilon \in E(R(\tau_i)) \setminus E(\tau_i)}(1+f_{\varepsilon}) &= \prod\limits_{e \in E(\tau)}f_e \prod\limits_{\varepsilon \in E(R(\tau)) \setminus E(\tau)}(1+f_{\varepsilon}). \label{eq:factor}
\end{align}
\end{proof}

\subsection{Factorisation of Integrals}
 We wish to show that the full weight of a tree also factorises as we require. 

The full weight of a tree is given by 
\be 
\tilde{\omega}(\tau)=\lim_{\Lambda \to \mathbb{R}^d}\frac{1}{|\Lambda|}\prod\limits_{i \in V(\tau)}\left(\int_{\Lambda} \mathrm{d} x_i \right) \omega(\tau).
\ee

Since we are assuming that the pair potentials are centred, the integrand depends only upon the distances between the particles. The thermodynamic limit is independent of boundary conditions and so we consider periodic boundary conditions. With periodic boundary conditions, we may make the change of variables from the spatial points to the differences. We may choose the change of variables to be $(x_1, \cdots , x_N) \mapsto (x_1, x_2-x_1, \cdots , x_N-x_1)$ and in this case all differences can be expressed in terms of the latter $N-1$ variables and so is independent of $x_1$ integral. The key point is the periodic boundary conditions, which allows for the change in coordinates to preserve the area over which we integrate.  
 
The factorisation of the integrals occurs analogously to the block factorisation of connected graph weights as explained in \cite{L04}. The graph $\mathcal{M}((V_i)_{i \in [k]})$ acts in the same way as the block cutpoint tree, where the vertices corresponding to trees bear the analogy with blocks and the vertices corresponding to the common label bear the analogy with cutpoints. First we give the definitions required to rigorously convey how to treat the factorisation of these integrals.

\begin{definition}[The Centre of a Tree]
The eccentricity of a vertex $v$ in a tree $\tau$ is: $\varepsilon(v)=\max \{d(v,w) | w \in V(\tau) \}$.

The radius of a tree is: radius$(\tau)=\min \{\varepsilon(v) | v \in V(\tau) \}$.

The centre of a tree is: $\{v \in V(\tau) | \varepsilon(v)=$ radius$(\tau) \}$.
\end{definition}
 \begin{remark} A priori the centre of a tree may be a set of cardinality strictly greater than $1$. We are working with a special class of trees for $\mathcal{M}((V_i)_{i \in [k]})$, which are bipartite and importantly have all their leaves in a single set. This is sufficient to get a unique centre.
\end{remark}
The idea is that the vertex $w$ for which the eccentricity of a vertex is attained is necessarily a leaf. If not then we always have a neighbour away from the original vertex which allows us to increase eccentricity by $1$. So any intersection point has odd eccentricity and any tree vertex has even eccentricity and so the centre can contain only vertices from one of the two sets in the bipartite graph. The centre cannot contain two points of distance at least two from each other and hence can only be a singleton, so the centre of $\mathcal{M}((V_i)_{i \in [k]})$ is well defined as a single point. 

When we want to understand the factorisation, we want a well-defined way of deciding what label to take as the redundant integral in each of the tree factors. When we see how to do this systematically, we achieve the factorisation of the tree weights for the tilde generating function as well.

In order to define the vertices over which we make our integrals independent, we define a digraph from the tree by orienting all edges away from the centre of the tree. In this case every vertex except the centre has a unique vertex pointing towards it. For each tree vertex outside of the centre, the intersection point neighbour that is the other endpoint of the edge pointing towards it is the vertex, which the integrand is taken independent of. Each intersection point outside of the centre has also precisely one incoming edge which marks the tree for which this point is integrated over. 

We now consider what happens at the centre. If the centre is an intersection point, then this is understood as the overall point which is not integrated over. If the centre is a tree point, then we may choose any of its neighbouring intersection points to be the independent point and we are left with this single independent point. 

Hence we have that the tilde weights factorise.

\subsection{Establishing the Combinatorial Relationship}
By combining equations \eqref{eq:scheme}, \eqref{eq:virial-Bell-expl}, and \eqref{eq:factor} we find:
\begin{align*}
\beta_{n+1}&= \sum_{k=1}^n \binom{-n}{k} k! \sum_{{(k_1,\ldots,k_n): k_i\geq 0}\atop{\sum_i k_i = k, \sum_i ik_i = n}} \frac{n!}{\prod_{i=1}^n k_i!} \prod_{i=1}^{n} \left(\frac{b_{i+1}}{i!}\right)^{k_i}\\
 &= \sum_{k=1}^n \binom{-n}{k} k! \sum_{{(k_1,\ldots,k_n): k_i\geq 0}\atop{\sum_i k_i = k, \sum_i ik_i = n}} \frac{n!}{\prod_{i=1}^n k_i! (i!)^{k_i}} \prod_{i=1}^{n} b_{i+1}^{k_i}\\
 &= \sum_{k=1}^n \binom{-n}{k} k! \sum_{{(k_1,\ldots,k_n): k_i\geq 0}\atop{\sum_i k_i = k, \sum_i ik_i = n}} \frac{n!}{\prod_{i=1}^n k_i! (i!)^{k_i}} \sum_{(\tau_1,\ldots,\tau_k)}\prod_{i=1}^k\tilde{\omega}(\tau_i),\end{align*}
where the sum is over $\tau_i$, which are spanning trees for the corresponding $b_i$ in the product on the previous line. The next point is to write this product of the weights of tree graphs as the weight of an $R$-faithful merging.  This involves assigning labels to the individual trees as described above. We then express it as:

 \be \beta_{n+1}= \sum_{k=1}^n \binom{-n}{k} k! \sum_{{(k_1,\ldots,k_n): k_i\geq 0}\atop{\sum_i k_i = k, \sum_i ik_i = n}} \frac{n!}{\prod_{i=1}^n k_i! (i!)^{k_i}} \sum\limits_{(\tau_1, \cdots, \tau_k)} \frac{1}{N(\tau_1, \cdots, \tau_k)} \sum_{\tau\in\mathcal{T}[n+1]}\tilde{\omega}(\tau),
\ee

where in this case $\tau$ is some $R$-faithful merging of $(\tau_1, \cdots, \tau_k)$ and $N(\tau_1, \cdots, \tau_k)$ is the number of such $R$-faithful mergings.

We want to express the latter sum as a sum on $\mathcal{T}[n+1]$ and drop the condition that the tree graph is an $R$-faithful merging of a given set of $k$ smaller tree graphs. Note that this can be done, since any $\tau\in\mathcal{T}[n+1]$ can be made its own $R$-faithful merging by choosing a suitable labeling on its vertices. However, if a tree graph $\tau$ has a splitting $\{\tau_1,\ldots,\tau_k\}$, then it also has splittings $\{\tilde{\tau}_1,\ldots,\tilde{\tau}_j\}$ with $j<k$ and with the $\tilde{\tau}_i$ being mergings of distinct $\tau_j$, meaning that some trees will appear multiple times in the last sum in the right-hand side. 

We can overcome this problem by rewriting this sum as follows. We call a tree graph $\tau$ $l$-splittable if $\tau$ has a splitting into $l$ parts, but not into $l+1$ parts. The number of $l$-splittable tree graphs is a combinatorial factor which depends on the partition scheme being used. Combining this factor with the other combinatorial factors yields:

\be
\beta_{n+1} = \sum_{l=1}^n C(n+1,l) \sum_{{\tau\in\mathcal{T}[n+1]:}\atop{\tau\textrm{ is $l$-splittable}}}\tilde{\omega}(\tau). \label{eq:virial-new}
\ee
where 
\be C(n+1,l)=\sum\limits_{k=1}^l {-n \choose k} k! \sum\limits_{\substack{(k_1, \cdots, k_n):k_i \geq 0 \\ \sum_ik_i=k, \sum_iik_i=n}}\frac{n!}{\prod_{i=1}^nk_i!(i!)^{k_i}}\sum\limits_{(\tau_1, \cdots, \tau_k)} \frac{1}{N(\tau_1, \cdots, \tau_k)}, \ee
with the sum over $(\tau_1, \cdots, \tau_k)$ being over splittings of $\tau$. 

\begin{theorem}
\label{thm:exactvirialtree}
The virial coefficients may be expressed in terms of the exponential generating function of non-separable tree graphs according to a partition scheme $R$ as follows:
\be
\frac{\beta_{n+1}}{(n+1)!}=[z^n] \frac{1}{n+1}\sum\limits_{l=1}^{\infty}q_{R,l}C(n+1,l) \frac{(T_{\mathrm{R},1, \tilde{\omega}}(z))^l}{l!},
\ee
where $q_{R,l}$ is a combinatorial factor denoting the number of ways to attach non-splittable graphs to make an $l$-splittable graph.
\end{theorem}

This theorem is a rewriting of \eqref{eq:virial-new}, where we realise that an $l$-splittable graph is uniquely defined by its $l$ non splittable components, which arise from the generating function $T_{\mathrm{R},1,\tilde{\omega}}(z)$ and the combinatorial factor $\frac{q_{R,l}}{l!}$ ensures each $l$-splittable graph is used only once.

\subsection{Obtaining Bounds from the Tree Partition Scheme}
\label{subsec:boundpartition}
The final stage is then to make bounds according to the assumptions made on the potential. The method of obtaining such a bound is to give a uniform upper bound for the product of $1+f$-factors and then to calculate the integral of the $|f|$-factors per each edge. The value of the integral is $C(\beta)$ as defined in section \ref{sec:clustertovirial}.

For stable potentials we have the following theorem.

\begin{theorem}
For $u:=\exp(2\beta B)$, we have the following bounds for the virial coefficients, when we have a stable potential.

\begin{align}
|\beta_{n+1}| &\leq C(\beta)^n u^{n-1}\sum\limits_{l=1}^n |C(n+1,l)| |\mathcal{T}_{\textrm{R},l}[n+1]| \label{eq:firststablebound} \\
\frac{|\beta_{n+1}|}{(n+1)!} &\leq [z^n]\frac{C(\beta)^n}{n+1}\sum\limits_{l=1}^{\infty} q_{\mathrm{R},l}|C(n+1,l)| \frac{\left(\frac{1}{u}T_{\mathrm{R},1}(uz)\right)^l}{l!}. \label{eq:secondstablebound}
\end{align}
\end{theorem}
This is a direct consequence of the following lemma and \eqref{eq:virial-new}, respectively theorem \ref{thm:exactvirialtree}  for \eqref{eq:firststablebound}, respectively \eqref{eq:secondstablebound}.
\begin{lemma}
There exists a permutation of the vertex labels $\Pi$ such that when applied to the labels in the product:
\be
\prod\limits_{\varepsilon \in E(R(\tau)) \setminus E(\tau)} (1+f_{\varepsilon}), \label{eq:stableproduct}
\ee
we have the upper bound $u^{n-1}$, for $\tau$ a tree on $n+1$ vertices.
\end{lemma}
\begin{proof}
We construct the permutation $\Pi$ required above, by using an important observation found in the book of Ruelle \cite{R69} about bounding $1+f$-factors.

Recall that stability requires that for any subset $S \subseteq [n+1]$, we have that:
\be \sum\limits_{\{i,j\} \in S^{(2)}}\Phi(x_i,x_j) \geq - \beta B |S|. \ee
This implies that
\be \prod\limits_{\{i,j\} \in S^{(2)}}(1+f_{i,j}) \leq \exp(\beta B |S|). \ee
We realise that the square of \eqref{eq:stableproduct} may written as:
\be \prod\limits_{i \in S} \prod\limits_{\substack{j \in S \\ j \neq i}} (1+f_{i,j}) \leq \exp(2 \beta B |S|), \ee
and hence that one of the factors in this product is less than $\exp(2 \beta B)$. We call this label $\iota_S$. Therefore we have:
\be \prod\limits_{\substack{j \in S \\ j \neq \iota_S}}(1+f_{\iota_S,j}) \leq \exp(2\beta B). \ee
In order to bound the product 
\begin{equation*}
\prod\limits_{\varepsilon\in E(R(\tau)) \setminus E(\tau)}(1+f_{\varepsilon}),
\end{equation*}
we consider the vertex labelled $n+1$ and its neighbours in the graph $G_{n+1}$, defined through its edge set $E(G_{n+1}):=E(R(\tau)) \setminus E(\tau)$ and call this neighbourhood $N(n+1)$. We let $S_{n+1}= N(n+1) \cup \{n+1\}$ and then relabel $n+1 \mapsto \iota_{S_{n+1}}$ and $\iota_{S_{n+1}} \mapsto n+1$. We call this transposition $\sigma_{n+1}$. 

This gives us that:
\be 
 \prod\limits_{\substack{j \in S_{n+1} \\ j \neq \sigma_{n+1}(n+1)}}(1+f_{\sigma_{n+1}(n+1),j}) \leq \exp(2 \beta B).
\ee

The permutation $\Pi$ is defined through a sequence of transpositions, which are defined inductively as below. 

For each stage in the induction, we retain the original labelling of the vertices in order to find which transposition is required to apply to obtain the bound. Given the relabellings for $[i+1,n+1]$, we understand how to relabel $i$. In the graph $G_i$ defined by edge set $E(G_i)=E(R(\tau)) \setminus E(\tau) \cap [i]^{(2)}$ and vertex set $[i]$, we consider the neighbourhood of $i$ and call this $N(i)$. We need to apply the sequence of transpositions generated so far to this set of 'original' labels in order to understand how they are labelled at the current stage. We use the notation $\pi S:= \{\pi(j) \vert \, j \in S\}$, where $\pi$ is a permutation and $S$ is a set.

We define 
\begin{equation*}
S_i :=\left(\prod\limits_{j=i+1}^{n+1} \sigma_j\right) (\{i\} \cup N(i)) = \sigma_{i+1} \cdots \sigma_{n+1} (\{i\} \cup N(i)).
\end{equation*}
We then define the transposition 
\begin{equation*}
\sigma_i =\left (\prod\limits_{j=i+1}^{n+1}\sigma_j(i)\, \iota_{S_i}\right).
\end{equation*}

The required final permutation is 
\begin{equation*}
\Pi = \prod\limits_{i=3}^{n+1} \sigma_i=\sigma_3 \cdots \sigma_{n+1},
\end{equation*}
which is what we apply to our labels. 

This process provides the identity:
\be 
\Pi \left(\prod\limits_{\varepsilon \in E(R(\tau)) \setminus R(\tau)}(1+f_{\varepsilon})\right) = \prod\limits_{i=3}^{n+1}\prod\limits_{\substack{j \in S_i \\ j \neq \iota S_i}}(1+f_{j, \iota_{S_i}}).
\ee

We apply this permutation $\Pi$ to the integrand since the final result is independent of the labelling of the particles.

Each factor in this product is bounded from above by $u:=\exp(2 \beta B)$ and we have $n-1$ of these. We thus have a factor $u^{n-1}$ in the virial expansion bound.
\end{proof}

 In the following section we illustrate this method for the Penrose partition scheme.
\section{The Penrose Partition Scheme}
\label{sec:penrosepartitionscheme}
In 1967, Penrose \cite{P67} introduced a partition scheme in the sense of Section \ref{sec:genpartitionschemes}. The scheme is as follows:

Suppose we have a graph $G$ with labels in $[n+1]$ and let $\tau \in \mathcal{T}[n+1]$.  For any vertex label $i$ of $\tau$, let $d(i)$ be the tree distance of the vertex labelled $i$ to $1$ and let $i'$ be the predecessor of $i$ i.e. $d(i')=d(i)-1$ and $\{i',i\} \in E(\tau)$. We associate to $\tau$, the graph $R_{\mathrm{Pen}}(\tau)$ found by adding (only once) to $\tau$ all edges $\{i,j\} \in E(\tau)$ such that either:
\begin{itemize}
\item[P1] $d(i)=d(j)$ (edges between vertices at same generation,)
\item[P2] $d(j)=d(i)-1$ and $i'<j$ (edges between vertices one generation away.)
\end{itemize}

The first result achieved is an expression for the virial coefficients in terms of the weighted trees.
\begin{theorem}
\label{thm:penroseexpansion}
The virial coefficients may be expressed in terms of weighted trees, with respect to the Penrose partition as follows.
\be
\frac{\beta_{n+1}}{(n+1)!} = \frac{1}{n+1} [z^n](1-T_{\mathrm{Pen},1, \tilde{\omega}}(z))^n, \label{eq:exactpenrose}
\ee
where
\be
 T_{\mathrm{Pen},1, \tilde{\omega}}(z)=\sum\limits_{n=1}^{\infty} \frac{z^n}{n!} \sum\limits_{\tau \in \mathcal{T}_{\mathrm{Pen},1}[n+1]}\tilde{\omega}(\tau).
\ee
\end{theorem}

As a consequence of the above writing of virial coefficients in terms of weighted trees, we obtain the following bounds for the virial coefficients and the radius of convergence of the virial expansion.

\begin{theorem}
\label{thm:stablebounds}
For a stable potential, the virial coefficients are bounded as:
\be
\frac{|\beta_{n+1}|}{(n+1)!} \leq \frac{C(\beta)^n}{n+1} \left(\frac{1+\frac{1}{u}T_{\mathrm{Pen},1}(uc(u))}{c(u)}\right)^n ,
\ee
where $c(u)$ is the smallest positive solution to
\be
uc(u)T_{\mathrm{Pen},1}'(uc(u))-T_{\mathrm{Pen},1}(uc(u))=u.
\ee
This provides the bound on the radius of convergence of the virial expansion as:
\be 
\mathcal{R}_{\mathrm{Vir}} \geq \frac{c(u)}{1+\frac{1}{u}T_{\mathrm{Pen},1}(uc(u))} C(\beta)^{-1}
\ee
\end{theorem}

\subsection{The Combinatorial Factors}
\begin{lemma}
The number of different Penrose-faithful mergings for a collection of trees $(\tau_1, \cdots, \tau_k)$ is:
\be
N(\tau_1, \cdots, \tau_k)= k! \frac{n!}{\prod_i k_i!(i!)^{k_i}}.
\ee
\end{lemma}
\begin{proof}

In order to assign labels to the individual trees in a monomial to obtain a Penrose-faithful merging, we first order the individual factors. This contributes $k!$ to the $N(\tau_1, \cdots, \tau_k)$. Once the trees are ordered we assign new labels to these trees. The trees all come with their original labels from the expression for connected graphs. We respect the ordering of the labels. We assign the label $1$ to the first tree, this is given to the vertex originally labelled $1$. We then assign $i$ labels to each tree with $i+1$ vertices. This is done in $\frac{n!}{\prod_i k_i!(i!)^{k_i}}$ ways. We assign these labels in their natural order to the vertices originally labelled in $[2,i+1]$. The vertex $1$ achieves a label from the preceding tree. From the preceding tree, we choose the vertex that is firstly the greatest distance from $1$ in the orignal tree and if there are more satisfying this property, we choose the smallest.

\end{proof}

\begin{figure}
$\begin{array}{c|c}
\begin{tikzpicture}
\tikzstyle{vertex}=[circle,draw=black,minimum size=17pt,inner sep=0pt]
	\node[vertex]	(a)	 at (0,0) {3};
	\node[vertex]	(b) at (1,2) {2};
	\node[vertex]	(c) at (1,4) {1};
	\node[vertex]	(d) at (2,0) {4};
	\node[vertex]	(e) at (3.5,0) {2};
	\node[vertex]	(f) at (4.5,2) {1};
	\node[vertex]	(g) at (5.5,0) {3};
	\node[vertex]	(h) at (7,0) {2};
	\node[vertex]	(i)	at (7,2) {1};
	
    \node[above] at (c.north) {1};
    \node[left] at (b.west) {2};
    \node[below] at (a.south) {3};
    \node[below] at (d.south) {4};
    \node[above] at (f.north) {4};
    \node[below] at (e.south) {5};
    \node[below] at (g.south) {6};
    \node[above] at (i.north) {6};
    \node[below] at (h.south) {7};	
	
	\draw (c)--(b);
	\draw (a)--(b);
	\draw (b)--(d);
	\draw (f)--(e);
	\draw (f)--(g);
	\draw (h)--(i);
	
	\draw[dashed] (a)--(d);
	\draw[dashed] (e)--(g);
\end{tikzpicture}
&
\begin{tikzpicture}
\tikzstyle{vertex}=[circle,draw=black,minimum size=17pt,inner sep=0pt]
	\node[vertex]	(a)	 at (1,4) {1};
	\node[vertex]	(b) at (1,3) {2};
	\node[vertex]	(c) at (0,2) {3};
	\node[vertex]	(d) at (2,2) {4};
	\node[vertex]	(e) at (1,1) {5};
	\node[vertex]	(f) at (3,1) {6};
	\node[vertex]	(g) at (3,0) {7};
	
	\draw (a)--(b);
	\draw (b)--(c);
	\draw (b)--(d);
	\draw (d)--(e);
	\draw (d)--(f);
	\draw (f)--(g);
	
	\draw[dashed] (c)--(d);
	\draw[dashed] (e)--(f);
\end{tikzpicture}
\end{array}$
\caption{An example of a Penrose-faithful merging. The labels in the vertices on the left-hand side are the original labels, and the labels by the vertices are the newly assigned labels. On the right-hand side we have the merging of these trees.}
\end{figure}
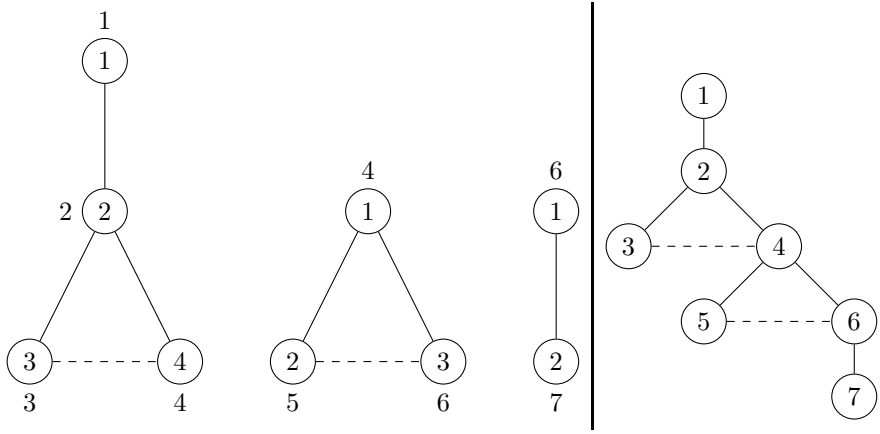

 We realise that for an $m$-splittable tree we have ${m-1 \choose k-1}$ ways of producing this tree from a product of $k$ terms. We have $m-1$ points at which to split the tree and we split at $k-1$ of these. The combinatorial factor in \eqref{eq:virial-new} is therefore:

\be C(n+1,m)=\sum\limits_{k=1}^m {-n \choose k} {m-1 \choose k-1}.\label{eq:penrosec} \ee

\begin{proposition}
We have the combinatorial identity:
\be 
\sum\limits_{k=1}^m {-n \choose k} {m-1 \choose k-1} = (-1)^m {n \choose m}. 
\ee
\end{proposition}
\begin{proof}
If we let $s=k-1$ in \eqref{eq:penrosec}, then we can recast the sum as:

\be \sum\limits_{s=0}^{m-1}{m-1 \choose m-1-s} {-n \choose s+1}. \ee

We realise this is the $z^{m-1}$ term of the generating function formed by multiplying $(1+z)^{m-1}$ and $\frac{1}{z}\left(\frac{1}{(1+z)^n}-1\right)$, which means we need to calculate:

\be  [z^{m}]\left(\frac{1}{(1+z)^{n-m+1}}-(1+z)^{m-1}\right)={-(n-m+1) \choose m}=(-1)^m{n \choose m} .\ee
\end{proof}

Hence, $C(n+1,m)=(-1)^m{n \choose m}$ and we have the expression:
\be \beta_{n+1} = \sum\limits_{m=1}^{n} (-1)^m {n \choose m} \sum\limits_{\tau \in \mathcal{T}_{\mathrm{Pen},m}[n+1]}\tilde{\omega}(\tau). \ee

Let $T_{\mathrm{Pen},1}(z)$ be the shifted exponential generating function for Penrose graphs that cannot be split.  It is defined by:

\be T_{\mathrm{Pen},1}(z):=\sum\limits_{n=2}^{\infty}|\mathcal{T}_{\mathrm{Pen},1}[n]|\frac{z^{n-1}}{(n-1)!}. \label{eq:T1def} \ee

The exponential generating function is shifted to make it easier to find the numbers $|\mathcal{T}_{\mathrm{Pen},m}[n]|$ from the above series.

To form an $m$-splittable tree, we need to attach $m$ non splittable trees together in the way described above. We have that
\be |\mathcal{T}_{\mathrm{Pen},m}[n]|=\sum\limits_{(k_1, \cdots, k_n)}\frac{(n-1)!m!}{\prod_{i=1}^nk_i!(i!)^{k_i}}\prod\limits_{i=1}^n|\mathcal{T}_{\mathrm{Pen},1}[i+1]|^{k_i}, \ee
where the sum above is over sequences $(k_1, \cdots, k_n)$ satisfying $\sum_i k_i=m$ and $\sum_i ik_i=n$. We can thus see this as the coefficient of $z^{n-1}$ in the product $T_{\mathrm{Pen},1}(z)^m$ and hence we obtain:
\be |\mathcal{T}_{\mathrm{Pen},m}[n]|=(n-1)![z^{n-1}](T_{\mathrm{Pen},1}(z))^m. \label{eq:T1Tm} \ee

We note that this identity is true for the weighted generating function $T_{\mathrm{Pen},1,\tilde{\omega}}(z)$ as well due to the factorisation of the weights and hence
\be
T_{\mathrm{Pen},m,\tilde{\omega}}(z)=(T_{\mathrm{Pen},1,\tilde{\omega}}(z))^m.
\ee

\subsection{The Generating Function $T_{\mathrm{Pen},1}(z)$}

The main result of this section is the number of non splittable trees with respect to the Penrose partition.

\begin{theorem}
\label{thm:noofnonseparable}
The number of non-splittable trees for the Penrose partition is:
\be
|\mathcal{T}_{\mathrm{Pen},1}[n+1]| = (n-1)^{n-1}, 
\ee
where we understand $0^0$ as $1$ to cover the case $n=1$.
\end{theorem}

In order to prove this, we first require modifications of the following two theorems.

\begin{theorem}[The Dissymmetry Theorem for Trees]
Let $T^{\bullet}(z)$ denote the exponential generating function of rooted trees (formed by taking the Euler derivative), $T^{-}(z)$ the collection of edge-rooted trees and $T^{\bullet - \circ}(z)$ the exponential generating function of oriented edge rooted trees. We have the identity due to Otter \cite{O48}
\be T^{\bullet}(z)+T^{-}(z)=T(z)+T^{\bullet - \circ}(z). \ee

\end{theorem}
We realise that rooting a tree at an edge is the same as splitting the vertex set into two sets either side of the edge and having a rooted tree for each half i.e. $\mathcal{T}^{-}=\mathcal{E}_2(\mathcal{T}^{\bullet})=\frac{(\mathcal{T}^{\bullet})^2}{2}$. We realise that if we have an oriented edge then we have $\mathcal{T}^{\bullet - \circ}=(\mathcal{T}^{\bullet})^2$. This provides us with the modified identity:
\be \mathcal{T}^{\bullet}(z)- \frac{1}{2}(\mathcal{T}^{\bullet}(z))^2 = \mathcal{T}(z). \label{eq:rootedtrees} \ee
\begin{corollary}
\label{cor:doublerooted}
We have the following relationship between doubly rooted and rooted trees:
\be T^{\bullet \bullet}(z)=\frac{T^{\bullet}(z)}{1-T^{\bullet}(z)}. \ee
\end{corollary}
\begin{proof}
If we take the Euler derivative of \eqref{eq:rootedtrees}, then we obtain:

\begin{align} T^{\bullet \bullet}(z)-T^{\bullet}(z)T^{\bullet \bullet}&=T^{\bullet}(z), \notag \\
T^{\bullet \bullet}(z)&=\frac{T^{\bullet}(z)}{1-T^{\bullet}(z)}. \end{align}
\end{proof}

\begin{theorem}[Functional Equation for Trees]
Rooted tree generating functions satisfy a recursion relation
\be T^{\bullet}(z)=z \exp(T^{\bullet}(z)) .\ee
In particular for the derivative of the tree generating function, we have:
\be T'(z)= \exp(z T'(z)). \label{eq:treerecursive} \ee
\end{theorem}
This identity is found in \cite{L04}.
\begin{proof} Consider a rooted tree. The root may have an arbitrary number of neighbours. One can view the rest of the tree outside of the root as a collection of trees rooted at the corresponding neighbour of the initial root. Hence, we have the first $z$ for the root of the tree and the exponential function since the order of the neighbours doesn't matter so we need to divide by $k!$ if we have $k$ neighbours and for each neighbour we have a rooted tree.
\end{proof}
\begin{corollary}
\label{cor:secondderivative}
We have that the second derivative of the tree generating function may be expressed in terms of the first:
\be
T''(z)=\frac{(T'(z))^2}{1-zT'(z)}.
\ee
\end{corollary}
\begin{proof}
If we differentiate \eqref{eq:treerecursive}, then we obtain:
\begin{align} T''(z) &=(T'(z)+zT''(z))\exp(zT'(z)) \notag \\
&=(T'(z))^2+zT'(z)T''(z), \notag \\
T''(z)(1-zT'(z)) &= (T'(z))^2, \notag \\
T''(z) &= \frac{(T'(z))^2}{1-zT'(z)}. \end{align}
\end{proof}

\begin{proof}[Proof of Theorem \ref{thm:noofnonseparable}]
We know that the sum of all tree graphs may be written in two ways using Cayley's formula:

\be \sum\limits_{m=1}^{n-1} |\mathcal{T}_{\mathrm{Pen},m}[n]|=n^{n-2}. \ee

We emphasise that any tree has a well defined $m$ for which is it $m$-splittable.

We can express the left hand side easily in terms of coefficients of generating functions from \eqref{eq:T1def} and \eqref{eq:T1Tm}. For the right hand side we want the number of trees on $n$ vertices to appear as the coefficient of $z^{n-1}$. This means we need to differentiate the tree generating function. Furthermore, the constant term of the left hand side is $0$, but $T'(0)=1$, so we must also take $1$ and obtain: 

\be (n-1)![z^{n-1}]\sum\limits_{m=1}^{n-1} (T_{\mathrm{Pen},1}(z))^m = (n-1)![z^{n-1}](T'(z)-1) \ee
which we have for every $n$. We also note that we may make the sum on the left hand side to infinity without affecting the outcome, since we will only have terms of greater powers of $z$ from adding these factors. Hence we have the identity:

\be \frac{T_{\mathrm{Pen},1}(z)}{1-T_{\mathrm{Pen},1}(z)}=T'(z)-1 .\ee

This can be manipulated into the form:

\be T_{\mathrm{Pen},1}(z)=1-\frac{1}{T'(z)}. \label{eq:T1inT} \ee

If we differentiate \eqref{eq:T1inT}, then we have the expression:

\be T_{\mathrm{Pen},1}'(z)=\frac{T''(z)}{T'(z)^2}. \ee

We use corollary \ref{cor:secondderivative} to obtain:
\be T_{\mathrm{Pen},1}'(z)=\frac{1}{1-T^{\bullet}(z)} \label{eq:penasrootedtree}, \ee

and using corollary \ref{cor:doublerooted}, we have:

\be T_{\mathrm{Pen},1}'(z)=T^{\bullet \bullet}(z)+1. \ee

Hence, we integrate both sides using the fact that $T_{\mathrm{Pen},1}(0)=0$ to obtain:
\be T_{\mathrm{Pen},1}(z)=\sum\limits_{n=1}^{\infty} n^n \frac{z^{n+1}}{(n+1)!} + z. \ee
\end{proof}

\subsection{Evaluating an Upper Bound}

For stable potentials, we have the bound, from section \ref{subsec:boundpartition}, 

\begin{align} \frac{|\beta_{n+1}|}{(n+1)!} &\leq \frac{C(\beta)^{n}}{n} u^{n-1} [z^n]\sum\limits_{m=1}^n {n \choose m} (T_{\mathrm{Pen},1}(z))^m \notag \\
&= \frac{C(\beta)^{n}}{n}[z^n]\left(1+\frac{1}{u}T_{\mathrm{Pen},1}(uz)\right)^n. \end{align}
We may represent this coefficient through the contour integral representation:
\be [z^n]\left(1+\frac{1}{u}T_{\mathrm{Pen},1}(uz)\right)^n = \frac{1}{2 \pi i}\oint_C \frac{\mathrm{d}\xi}{\xi^{n+1}}\left(1+\frac{1}{u}T_{\mathrm{Pen},1}(u\xi)\right)^n \ee
where $C$ is a contour around $0$ within distance $\frac{1}{eu}$ from the origin, which is the radius of convergence of $T_{\mathrm{Pen},1}(uz)$. The right hand side can be bounded above by taking a contour with fixed radius $R$ and since all coefficients are positive in $T_{\mathrm{Pen},1}(uz)$ the maximum value of the modulus along this circle is the function evaluated at $R$. We thus have:

\be
[z^n]\left(1+\frac{1}{u}T_{\mathrm{Pen},1}(uz)\right)^n \leq \left(\frac{1+\frac{1}{u}T_{\mathrm{Pen},1}(uR)}{R}\right)^n
\ee
for positive $R < \frac{1}{eu}$.

In order to find an optimal $R$, we differentiate and find the smallest positive root of the derivative equal to zero. That is we need to solve:
\be 
\frac{(1+\frac{1}{u}T_{\mathrm{Pen},1}(uR))^{n-1}n(RT_{\mathrm{Pen},1}'(uR)-1-\frac{1}{u}T_{\mathrm{Pen},1}(uR))}{R^{n+1}}=0.
\ee
Since $\frac{1}{u}T_{\mathrm{Pen},1}(uR)$ cannot equal $-1$, it follows that the equation is solved by:
\begin{align*}
 RT_{\mathrm{Pen},1}'(uR)-1-\frac{1}{u}T_{\mathrm{Pen},1}(uR)&=0, \\
uRT_{\mathrm{Pen},1}'(uR)-T_{\mathrm{Pen},1}(uR) &=u. \end{align*}
and hence we have the smallest positive solution to this equation dependent on $u$ and call this $c(u)$. This provides the upper bound for virial coefficients:
\be 
\frac{|\beta_{n+1}|}{(n+1)!} \leq \frac{C(\beta)^n}{(n+1)!}\left(\frac{1+\frac{1}{u}T_{\mathrm{Pen},1}(uc(u))}{c(u)} \right)^n .
\ee
This leads to the lower bound on the radius of convergence of the virial expansion:
\be
\mathcal{R}_{\mathrm{Vir}} \geq \frac{c(u)}{1+\frac{1}{u}T_{\mathrm{Pen},1}(uc(u))}C(\beta)^{-1}. \ee

\section{Conclusions and Future Work}
\label{sec:conclusion}
Combining the Penrose tree partition scheme with the algebraic relationship between the virial and the cluster coefficients, we have been able to rigorously achieve Groeneveld's lower bounds on the radius of convergence of the virial expansion for a classical gas with purely repulsive interactions. This may also be generalised to stable potentials through using a relabelling of the trees. This has also provided representation of the virial coefficients in terms of weighted trees.  

Graph tree equalities, through the fundamental theorem of calculus of Brydges and Federbush \cite{BF78}, are an alternative way of understanding rewriting connected graphs in terms of trees. These have been generalised to matroids by Faris \cite{F12} and given a symmetric form by Abdesselam and Rivasseau \cite{AR94}. These identities have not yet been made amenable to finding a product structure. A further possibility to deal with stable potentials is to use the ideas of Rivasseau \cite{RT14} where a Kruskal algorithm gives an optimal way of writing a tree partition.

Future work is to understand what one can achieve from other partition schemes and whether there is any sense of an optimal partition scheme from which one can obtain bounds. It is anticipated that future papers will understand the application of these to hardcore potentials in the discrete and continuous case, where the new tree weights should assist in making more accurate estimates. In addition, this new expression for virial coefficients in terms of trees should provide an easier method of making calculations than the method of Ree-Hoover bounds \cite{RH64}.

\medskip\noindent
{\bf Acknowledgements.}
S. R. would like to acknowledge financial support from NWO grant 613.001.015. S. J. T. would like to acknowledge financial support from EPSRC grant EP/G056390/1, the London Mathematical Society Postdoctoral Mobility Grant, the SFB TR12 and the EPSRC grant EP/L025159/1. The authors would like to thank Roberto Fern\'{a}ndez for his helpful comments in improving this article. Furthermore, the authors would like to thank the hospitality at Ruhr-Universit\"{a}t Bochum, Universiteit Leiden, Universiteit Utrecht and University of Warwick. S. J. T. would also like to thank David Brydges for other helpful comments on the article.

\end{document}